\documentclass[12pt]{article}

\usepackage[T1]{fontenc}
\usepackage[utf8]{inputenc} 

\usepackage{amscd,amsfonts,amssymb,amsthm,amscd,amsmath}
\usepackage{bm}
\usepackage{caption}
\usepackage[mathscr]{eucal} 
 \usepackage{latexsym}
\usepackage{mathptmx}
\usepackage{mathrsfs} 
\usepackage{mathtools}

\usepackage{color}
\usepackage{cite}
\usepackage{epsfig}
\usepackage{graphicx}
\usepackage{bm}

\newtheorem{theorem}{Theorem}[section]
\newtheorem*{definition}{Definition}
\newtheorem{lemma}[theorem]{Lemma}
\newtheorem{proposition}[theorem]{Proposition}

\numberwithin{equation}{section}

\def\cB{{\cal B}}

\def\cD{{\cal D}}

\def\cG{{\cal G}}
\def\cH{{\cal H}}

\def\cK{{\cal K}}
\def\cL{{\cal L}}
\def\cM{{\cal M}}
\def\cN{{\cal N}}
\def\cO{{\cal O}}

\def\cV{{\cal V}}

\def\fA{{\mathfrak A}}

\def\fR{{\mathfrak R}}
\def\fS{{\mathfrak S}}

\def\RR{{\mathbb R}}

\newcommand{\bv}{{\bm v}}

\def\eg{\textit{e.g.\ }}

\def\ie{\textit{i.e.\ }}

\newcommand{\ad}[1]{\mbox{ad}  \, #1}

\newcommand{\be}{\begin{equation}}
\newcommand{\ee}{\end{equation}}

\begin{document}

\title{\hspace*{-2mm}
  \mbox{Arrow of time and quantum physics} \\[4mm]
\large Dedicated to Roberto Longo on the occasion of his 70th birthday 
}
\author{\large Detlev Buchholz${}^{(1)}$ \ and 
\ Klaus Fredenhagen${}^{(2)}$ \\[5mm]
\small 
${}^{(1)}$ Mathematisches Institut, Universit\"at G\"ottingen, \\
\small Bunsenstr.\ 3-5, 37073 G\"ottingen, Germany\\[5pt]
\small
${}^{(2)}$ 
II. Institut f\"ur Theoretische Physik, Universit\"at Hamburg \\
\small Luruper Chaussee 149, 22761 Hamburg, Germany \\
}
\date{}

\maketitle

\begin{abstract}
  \noindent Based on the hypothesis that the
  (non-reversible) arrow
  of time is intrinsic in any system, no matter
  how small, the consequences
  are discussed. Within the framework of local quantum physics 
  it is shown how such a semi-group action of time can consistently
  be extended to that of the group of spacetime translations
  in Minkowski space. In presence of massless excitations, 
  however, there arise ambiguities in
  the theoretical extensions of the 
  time translations to the past. The corresponding 
  loss of quantum information on states upon time is determined.
  Finally, it is explained how the description of operations in
  classical terms combined with constraints
  imposed by the arrow of time leads to a quantum theoretical framework.
  These results suggest that the arrow of time is 
  fundamental in nature and not merely a consequence of statistical effects 
  on which the Second Law is based. 
\end{abstract}

\newpage
\section{Introduction}
\label{sec1}
\setcounter{equation}{0}

The \textit{arrow of time} is a subject of continuing discussions
ever since this
term was coined by Eddington \cite{Ed}. In brief, this topic can be
described as follows: the time parameter that enters into the fundamental
equations of physics can be reversed, which in principle seems to
allow physical systems to move backwards in time.
On the other hand, there is overwhelming evidence that this does not
happen in the real world. The standard resolution of this
apparent clash between
theory and reality is based on the argument that such time reversed
processes are exceedingly unlikely (Second Law). Therefore, they were
and will never be observed. 

It is the aim of the present article to propose an alternative view on
this topic.  We will discuss the hypothesis that the direction of time
is inherent in all systems, no matter how small. The evolution of
time thus has to be described by a semigroup, there is no inverse.
More concretely, consider an observer located
at a particular moment at point $o$ in
Minkowski space $\cM$.
As time proceeds, such observers and all their offsprings will enter
the forward lightcone $\cV_o \coloneqq o + V_+ \subset \cM$,
where $V_+$ is the semigroup of all positive timelike translations.
Observers at $o$ can neither re-enter the
past lightcone $\cV_{o \, -}$, 
where they came from, nor the spacelike complement of $o$;
it is  also not possible for them to send any 
instruments there. We denote this non-accessible region by $\cN_o$.

Information obtained in the past is encoded in material bodies
that accompany any observer. It is laid down in books or stored in other
devices, not least in our brains. They contain information about past
events, observations, experiments, the resulting data, and the
theories developed on their basis. The amount of this kind of
information grows steadily with time. Irrespective of their truth
value, these informations can be treated as factual
and described in classical terms (ordinary language).
It is impossible to go backwards in time
in order to verify their initial quantum features.
Any setup for future experiments is described in such classical
terms if they have been proven to work in the past. In 
experiments quantum effects can be explored, re-examined and
confirmed. But with the registration of the results these
become part of the past and therefore facts in the above sense.

Local quantum physics \cite{Ha} allows one to make statistical predictions
for future measuring results. It is even more important 
in the present context that, on the basis of theory, one can return to the past 
and develop meaningful scenarios of what has happened there and
elsewhere. It turns out, however, that back-calculations in
time do not allow unambiguous statements about the past in general.
These facts suggest to explore in more detail
the consequences of the hypothesis
that time has no inverse. There appear the following questions:
\begin{itemize}
\item[(I)] Is the hypothesis of an intrinsic arrow of time compatible
  with the successful theoretical treatment of time as a group?
\item[(II)] What are the uncertainties in the theoretical
  description of the past that arise from this hypothesis 
  and how do they manifest themselves?
\item[(III)] How big is the 
  loss of information on the properties of 
  states over time that arises from these uncertainties? 
\item[(IV)] Does the arrow of time enforce the quantum 
  features of operations that are described by classical
  concepts?
\end{itemize}    

In the present article we will provide answers to these questions.
Our arguments are based on previous work 
and some measure for the information contained in 
states that was recently invented by Roberto Longo
et al.\ \cite{CiLoRaRu}. We will be led
to the conclusion that the hypothesis of a fundamental arrow
of time is not only meaningful but also leads to the resolution
of some theoretical puzzles. 

Our article is organized as follows.
In Sect.\ 2, we introduce our framework 
and recall from \cite{BuRo} how to identify ground states (vacua) by
observables in a given lightcone $\cV_o$, together
with the semigroup $V_+$ of
time translations. Any such state leads to a unitary representation
of the semigroup. It can be extended to a representation
of the group of spacetime translations on $\cM$ that allows one to 
perform computations backwards in time. Section 3 contains a
discussion of the ambiguities concerning the past
which arise in these extensions and a proof that
they are related to states of arbitrarily small mass. 
In Sect.\ 4 we use the measure of information
introduced in \cite{CiLoRaRu} and show that the information which
can be extracted from such states by observables in lightcones
decreases monotonically in time. 
We then recall in Sect.~5 an approach to local quantum 
physics which proceeds from the concept of
operations \cite{BuFr2}. These operations are described in terms
of classical physics and are subject to causal relations  
that incorporate the arrow of time. Without imposing any
quantization rules from the outset, this approach gives rise
to genuine quantum theories. It shows that 
the arrow of time entails the non-commutative structures of
quantum physics. So this arrow may be regarded as a fundamental
property, complementing its statistical explanations,
based on the second law. 


\section{ \hspace*{-5mm} Time as a semigroup and its extension to a group}
\label{sec2}
\setcounter{equation}{0}

In this section we introduce our framework and explain 
the construction of a unitary representation
of the semigroup $V_+$ which induces the action
of time translations on the observables in a
given lightcone. We then exhibit an extension of this 
representation to the group of all spacetime translations
$\RR^4$. It determines lightcone algebras all over 
Minkowski space and acts covariantly (geometrically) on them.
This provides an affirmative answer to question~(I). 

In order to justify our input let us begin with
a brief remark. At each moment,
there exists a multitude of observers at points
$o_1, \dots, o_n \in \cM$ who follow their momentary 
arrow of time, depending on their velocity. The union of
their forward light cones is contained in some bigger
lightcone with apex $o$ in their pasts. If they
agree on this point, they can set their clocks
to $0$ there. For the present discussion it implies that
it suffices to consider a single lightcone
${\mathcal V}_o \coloneqq o + V_+$~and~the~semigroup
of time translations
$V_+ =
\{ \tau \coloneqq (t,t \, \bv) : t \geq 0, \, | \bv | < 1 \}$ acting on it. 

Let $\cV_o$ be given and let $\cV \mapsto \fA(\cV)$
be the respective (isotonous) net of lightcone subalgebras, 
$\cV \subset \cV_0 \,$; their local substructure will be used 
later. We assume that these algebras are
unital C*-algebras. The action of the time translations
$\tau \in V_+$ on these algebras is induced by morphisms $\alpha_\tau$,
\be \label{e.2.1} 
\alpha_\tau (\fA(\cV_o)) = \fA(\cV_{o + \tau})
\subset  \fA(\cV_o) \, , \quad \tau \in V_+ \, .
\ee
Composing them yields 
$\alpha_{\tau_1} \alpha_{\tau_2} = \alpha_{\tau_1 + \tau_2}$
for $\tau_1, \tau_2 \in V_+$. 

Turning to the states, experience shows that
ensembles with given properties
can repeatedly be prepared. It suggests that there is some
stationary background state which permits these
operations, \eg an equilibrium state 
or a ground state. We discuss here the latter scenario  
and use the following definition, where assumptions
made in \cite[Sect.\ 3]{BuRo} are slightly weakened. 
\begin{definition}
  Let $\fA(\cV_o)$ be given. A state $\omega_{ 0}$ on this algebra
  is said to be a ground state (vacuum) for all inertial observers in
  $\cV_o$ if it satisfies the following conditions. \\[-10mm]
  \begin{itemize}
  \item[(a)] $\omega_0 \, \alpha_\tau = \omega_0$ \ for \ $\tau \in V_+$. 
  \item[(b)] $\tau \mapsto \omega_0(A^* \alpha_\tau(B))$
    is continuous for $A,B \in \fA(\cV_o)$.
  \item[(c)] The functions   $\tau \mapsto \omega_0(A^* \alpha_\tau(B))$ 
    extend continuously to the complex domain $V_+ + i \, V_+$ and are
    analytic in its interior. Their modulus is bounded on this
    domain by $\sqrt{\omega_0(A^*A) \omega_0(B^*B)}$
    for $A,B \in \fA(\cV_o)$. 
  \end{itemize}  
\end{definition} 
\noindent \textbf{Remark:} These properties can in principle be tested
by observers in $\cV_o$. 

Vacuum states on $\fA(\cV_o)$
determine a continuous unitary representation of the semigroup
$V_+$ in the corresponding GNS-representation. It extends to
a representation of the spacetime translations $\RR^4$
on the net of all lightcone algebras. 
We add here to results given in \cite[Sect.\ 3]{BuRo}.

\begin{proposition} \label{p.2.1}
  Let $\fA(\cV_o)$ be given, let $\omega_0$ be a vacuum state on
  this algebra, and let $(\pi_0, \Omega_0, \cH_0)$ be the
  corresponding GNS-representation. 
  \begin{itemize}
  \item[(i)] The vector $\Omega_0$ is cyclic for each
    algebra $\pi_0(\fA_0(\cV_{o + \tau}))$, $\tau \in V_+$. 
  \item[(ii)] There exists a continuous unitary representation
    $\tau \mapsto U_0(\tau)$ 
    on $\cH_0$ that implements the action of the
    semigroup $V_+$ on $\fA(\cV_o)$,
    \be \label{e.2.2} 
    \ad{U_0(\tau)} ( \pi_0(A) ) =  \pi_0(\alpha_\tau(A)) \, , \quad
    A \in \fA(\cV_o) \, .
    \ee
    It leaves the representing vector invariant,
    $U_0(\tau) \Omega_0 = \Omega_0$, $ \tau \in V_+$. This
    unitary representation is unique. 
  \item[(iii)] The representation $U_0$ of $V_+$ can be extended to
    a continuous unitary representation $U$ of the
    group $\RR^4$ of spacetime translations on $\cM$. Its adjoint
    action on the given algebra
    defines a net of lightcone algebras on $\cM$
    on which it acts covariantly (geometrically). 
    Moreover, $U$ satisfies the 
    relativistic spectrum condition and leaves
    the vector $\Omega_0$ invariant. 
  \end{itemize}  
  \end{proposition}  
\begin{proof}
(i) According to property~(c) of vacuum states,   
the vector-valued functions $\tau \mapsto \pi_0(\alpha_\tau(A)) \Omega_0$ 
extend continuously to 
functions on the complex domain $V_+ + i\, V_+$
that are analytic in its interior, $A \in \fA(\cV_0)$.
Now if $\Psi \in {\cal H}_0$ is a vector
in the orthogonal complement of
$\pi_0(\alpha_\sigma(\fA(\cV_o)) \Omega_0$
for some $\sigma \in V_+$, it follows from the isotony 
of the net of lightcone algebras
and the covariant action of the time translations that 
$(\Psi, \pi_0(\alpha_\tau(A)) \Omega_0) = 0$ for 
all $\tau \in \sigma + V_+$. The edge-of-the-wedge theorem then implies
$(\Psi, \pi_0(\alpha_\tau(A)) \Omega_0) = 0$ for all $\tau \in V_+$
and hence $(\Psi, \pi_0(A) \Omega_0) = 0$, $A \in \fA(\cV_o)$.
Whence $\Psi = 0$ 
since the GNS vector $\Omega_0$ is cyclic for $\pi_0(\fA(\cV_0))$, 
proving the first statement. 

\noindent (ii)
Making use of property (a) of vacuum states, one can consistently
define isometries $U_0(\tau)$, $\tau \in V_+$, on $\cH_0$, putting
\be \label{e.2.3} 
U_0(\tau) \, \pi_0(A) \Omega_0 \coloneqq \pi_0(\alpha_\tau(A)) \Omega_0 \, ,
\quad A \in \fA(\cV_o) \, .
\ee
By the preceding step, these isometries have a dense
range in $\cH_0$ and hence are unitary. So they induce the
endomorphic action of $V_+$ on the lightcone algebras and
leave the vector $\Omega_0$ invariant. This fixes them uniquely. 
Moreover,  they are weakly continuous according to property (b)
of vacuum states and satisfy 
\be \label{e.2.4} 
U_0(\tau_1) U_0(\tau_2) = U_0(\tau_1 + \tau_2) =
U_0(\tau_2) U_0(\tau_1) \, , \quad \tau_1, \tau_2 \in V_+ \, .
\ee

\noindent (iii) 
For the proof that $U_0$ can be extended to all spacetime translations,
we make use of the fact that any $x \in \RR^4$ can be presented as
difference of elements of $V_+$. 
 So let $x = \tau_1 - \tau_2 = \tau_3 - \tau_4$,  hence
$ \tau_1 + \tau_4 = \tau_2 + \tau_3$. Making use of equation \eqref{e.2.4}
it follows that
$U_0(\tau_1) U_0(\tau_2)^{-1} =  U_0(\tau_3) U_0(\tau_4)^{-1}$.
So the operators
\be \label{e.2.5}
U(x) \coloneqq U_0(\tau_1) U_0(\tau_2)^{-1} \, , \quad x = \tau_1 - \tau_2 \in 
\RR^4  \,  ,
\ee
are well defined, unitary, and they extend $U_0$. 
By a similar computation one finds that
$U(x) U(y) = U(x + y)$ for $x,y \in \RR^4$. We put now
\be \label{e.2.6}
\fA_0(\cV_{o + x}) \coloneqq
\ad{U(x)}(\pi_0(\fA(\cV_o))) \, , \quad x \in \RR^4 \, , 
\ee
hence
$\fA_0(\cV_{o + \tau}) = \pi_0(\fA(\cV_{o + \tau}))$, $\tau \in V_+$. 
The covariant action of $U$ on these algebras follows from its definition.
Since the translations act transitively on the lightcones
in Mink\-owski space, one obtains all corresponding algebras. 
Next, let $y - x \in V_+$. There is some $\tau \in V_+$ such that
$(x + \tau), (y + \tau) \in V_+$, hence
 \begin{align} \label{e.2.7}
&  \fA_0(\cV_{o + y}) = \ad{U(\tau)^{-1}}
  \ad{U(y +\tau)} (\fA_0(\cV_o)) =
  \ad{U(\tau)^{-1}} \pi_0(\fA(\cV_{o + y + \tau})) \\
&  \subset \ad{U(\tau)^{-1}} \pi_0(\fA(\cV_{o + x + \tau}))
  = \ad{U(\tau)^{-1}} \ad{U(x +\tau)} (\fA_0(\cV_o)) =
    \fA_0(\cV_{o + x}) \, . \nonumber
 \end{align}
So the resulting family of lightcone algebras is isotonous, \ie it constitutes 
a net.

To verify the continuity and spectral properties of $U$,  
let \mbox{$A,B \in \fA(\cV_o)$} and let $x + \tau \in V_+$ for 
suitable $\tau \in V_+$. Property (a) of vacuum states implies that
\be \label{e.2.8}
\langle \pi_0(A) \Omega_0 , U(x)  \pi_0(B) \Omega_0 \rangle
=   \langle \pi_0(\alpha_\tau(A)) \Omega_0,
 \pi_0(\alpha_{\tau + x}(B)) \Omega_0 \rangle \, .
\ee
Thus, by property (c), the function
$x \mapsto \langle \pi_0(A) \Omega_0 , U(x)  \pi_0(B) \Omega_0 \rangle$
is continuous and can analytically be continued into the domain
$\RR^4 + i \, V_+$. Moreover, its modulus is bounded there
by $\| \pi_0(A) \Omega_0 \| \|  \pi_0(B) \Omega_0 \|$.  It then 
follows from standard arguments in the theory of Laplace
transforms that the spectrum of $U$ is
contained in the closed forward light cone $\overline{V}_+$.
The invariance of $\Omega_0$ under the action
of $U$ follows from equations \eqref{e.2.5} and  \eqref{e.2.3}, 
completing the proof.
\end{proof}  

This result shows that the hypothesis of a fundamental arrow
of time leads, under meaningful assumptions, to the
theoretical description of arbitrary spacetime translations,
forming a group and acting covariantly on observables 
all over Minkowski space. So that hypothesis is compatible with the 
common assumptions made in theoretical physics. However, as we
shall see in the next section, this theoretical extension of the
semigroup of time to a group is in general not
unique. Statements about the past then involve unavoidable
uncertainties.

\section{The uncertain past}
\label{sec3}
\setcounter{equation}{0}

We exhibit now possible ambiguities arising in the
extension of the semigroup of time translations to a group,
constructed in the preceding section, 
and relate them to specific properties of the energy-momentum
spectrum. The results provide answers to question (II), raised above.
In this analysis we make use of the net of local subalgebras
contained in the given lightcone algebra, 
$\cO \mapsto \fA(\cO) \subset \fA(\cV_o)$, whose
elements are assumed to  commute at spacelike distances.
In representations induced by a vacuum state, 
this net extends to a local net $\cO \mapsto \fA_0(\cO)$
on Minkowski space~$\cM$, obtained by the adjoint action of the
unitary group $U$, cf.\ equation \eqref{e.2.7}.
Let $\fR$ be the von Neumann algebra
on $\cH_0$, which is generated by this net. By a result of
Araki on vacuum representations, cf.\ \cite[Sec.~2.4]{Sa}, 
its commutant coincides with the center of 
the algebra, \ie $\fR' \subset \fR$, 
and it is pointwise invariant under the adjoint action
of $U$. Thus, by central decomposition, we may assume that
the multiples of $\Omega_0$ are the only 
$U$-invariant vectors in $\cH_0$ and that
$\fR$ coincides with the algebra $\cB(\cH_0)$
of all bounded operators on $\cH_0$. 
Moreover, the spectrum of~$U$ is a Lorentz invariant
subset of the lightcone $\overline{V}_+$ in
momentum space, cf.\  \cite{Bo1}.

Whereas the theoretical extension of the algebra $\fA(\cV_o)$ to
all of Minkowski space, obtained in this manner, comprises maximal 
information, this extension is in general not unique. 
Whenever this happens, the spectrum of $U$ fills the whole lightcone,
\ie there exist excitations of arbitrarily small mass.
In the proof of this assertion we make use of the following lemma,
where the weak closure of the lightcone algebra is denoted by
$\fR(\cV_o) \coloneqq \pi_0(\fA(\cV_o))^-$.
\begin{lemma} \label{l.3.1}
  Let $Z \in \fR(\cV_o)$ be the largest projection which annihilates
  the vacuum, $Z \, \Omega_0 = 0$. Then \ $\ad{U(x)}(Z) = Z$ \
  for all $x \in \RR^4$. 
\end{lemma}  
\begin{proof} Let $x \in \RR^4$. Because of the
  covariant action of $U$ on the lightcone algebras,
  the operator $Z_x \coloneqq \ad{U(x)}(Z)$ is the largest projection in
  $\fR(\cV_{o + x})$ which annihilates $\Omega_0$.
  If $(y - x) \in V_+$ it follows that
  $\fR(\cV_{o + y}) \subset \fR(\cV_{o + x})$ and consequently
  $Z_y \leq Z_x$. 

  We pick now a one-parameter family of time translations
  $t \mapsto \tau(t) = (t, t \bv)$ for fixed $\bv$. 
  By the preceding argument, the projections
  $Z_{\tau(t)} = \ad{U(\tau(t))}(Z)$, $t \in \RR$, commute and generate,
  together with $1$, an
  abelian von Neumann algebra. Because of the spectral properties
  of $U$, this implies by a theorem of
  Borchers that $Z_{\tau(t)} = Z$, $t \in \RR$, for any
  admissible choice of $\bv$, cf.\ \cite[Thm.\ 2.4.3]{Sa}.
  The statement then follows. 
\end{proof}

\vspace*{-3mm}
The projection $Z$ in the preceding proposition encodes
information about the degree of ambiguity involved in the
extension of time translations to the past.
Depending on the underlying theory, there occur the
following possibilities. \\[2mm]
(a) \ $Z = (1 - P_{\Omega_0})$, where $ P_{\Omega_0}$ is the one-dimensional
projection onto the vacuum vector. Since $1-Z$ is an element of 
$\fR(\cV_0)$, this algebra coincides with $\cB(\cH_0)$,
the unitaries $U \in \cB(\cH_0)$ are uniquely fixed, 
and all lightcone algebras coincide.
This possibility occurs in theories where the
spectrum of $U$ has a gap between the point $0$, corresponding
to the vacuum, and the rest of the spectrum \cite{SaWo}.  \\[2mm]
(b) \ $ 0 < Z < (1 - P_{\Omega_0})$. Then $\fR(\cV_o)$ is
a proper subalgebra of $\cB(\cH_0)$ and thus has a non-trivial
commutant $\fR(\cV_o)'$. Disregarding multiples of $1$,
the operators in $\fR(\cV_o)'$ do not commute with the
translations $U$; otherwise, $\Omega_0$ 
would not be unique. Now let $W \in \fR(\cV_o)'$ be unitary. 
Then $x \mapsto U_W(x) \coloneqq \ad{W}(U(x))$
is another unitary representation of the
translations~$\RR^4$ whose adjoint action on $\fR(\cV_o)$ for
time translations $\tau \in V_+$ coincides with those of the
unique initial $U_0$. More generally,
one can modify $U$ by cocycles with values in $\fR(\cV_o)'$.
It shows that the past and spacelike complement $\cN_o$ of the spacetime 
point $o$ can be described in different ways without modifying any 
observations in the future $\cV_o$. \\[2mm]
(c) \ $Z = 0$. Then $\Omega_0$ is cyclic and separating for
$\fR(\cV_o)$. By modular theory, its 
commutant $\fR(\cV_o)'$ is anti-isomorphic to $\fR(\cV_o)$.
As in (b), $U$ is not fixed, which has the same consequences.
This special case occurs in asymptotically complete theories describing
exclusively massless particles, cf.\ \cite[Prop.\ 4.2]{Bu2}. 

Thus, apart from case (a), one is faced with  
theoretical uncertainties involved in back-calculations to
the past. As was mentioned, such 
uncertainties appear in theories of massless
particles. We show next that in cases
(b) and (c) the spectrum of 
$U$ never has a gap between the vacuum and
the excited states.
In the proof we rely on a
fundamental result by Borchers on
modular inclusions \cite{Bo2, Fl}.

\begin{proposition}
  If the extension $U$ of the unitary time translations
  $U_0$ on $\fA_0(\cV_o)$ is not unique, its spectrum 
  consists of the closed cone $\overline{V}_+$.  In
  particular, it contains contributions with arbitrarily
  small mass. 
\end{proposition}

\vspace*{-4mm}
\begin{proof} As was shown above,
  the extension $U$ is not unique iff 
  $(1 - Z) > P_{\Omega_0}$.
  By the definition of $Z$, any positive operator
  $A \in \fR(\cV_o)$ that annihilates $\Omega_0$ is
  dominated by a multiple of $Z$, \ie $ 0 \leq A \leq \|A \| \, Z$. 
  Hence $\Omega_0$ is separating for   
  the algebra $\fS(\cV_o) \coloneqq (1-Z) \, \fR(\cV_0) \, (1 - Z)
  \subset \fR(\cV_o)$. It is also clear that 
  $\fS(\cV_o) \, \Omega_0$ is dense in $(1 - Z) \cH_0$.
  Let $\Delta_0$ be the modular operator
  on $(1 - Z) \cH_0$ fixed by $(\fS(\cV_o), \Omega_0)$. 
  Since $U$ commutes with $Z$, it leaves $(1 - Z)\cH_0$
  invariant and 
  $\ad{U(\tau)}(\fS(\cV_o)) \subset \fS(\cV_o)$ for
  $\tau \in V_+$. So in view of the spectral properties
  of $U$ one obtains the equality, cf. \cite{Bo2,Fl}, 
  \be \label{e.3.1}
  \ad{\Delta_0^{i\sigma}}(U(x)) = U(e^{-2 \pi \sigma} x) \, , \quad
  \sigma \in \RR \, , \ x \in \RR^4 \, .
  \ee 
  Hence the spectrum of $U$ on $(1 - Z) \cH_0 $ is
  invariant under scale transformations. Since 
  $( 1 - Z) \cH_0$ contains, apart from multiples of $\Omega_0$, only 
  vectors which are not invariant under the action of $U$, 
  the spectrum of $U$ includes a ray. So 
  the statement follows from the fact
  that the spectrum is Lorentz invariant.  
  \end{proof}  

\vspace*{-8mm}
\section{Loss of quantum information}
\label{sec4}
\setcounter{equation}{0}
As we have seen, the theoretical uncertainties about the past 
manifest themselves in specific spectral properties of
the unitaries $U$ on the subspace $(1-Z) \cH_0$.
We will now answer question (III) above and 
determine the corresponding losses of information about these 
states over time,
disregarding the vacuum state $\Omega_0$ which is stationary. 
Here we rely on definitions and results in \cite{CiLoRaRu}.

For the analysis of the states in $(1-Z) \cH_0$ it suffices to consider
the subalgebras 
$\fS(\cV_{o + \tau}) \coloneqq
(1-Z) \, \fR(\cV_{o + \tau}) \, (1-Z)$, $\tau \in V_+$.
Putting $\cK_{0} \coloneqq (1 - Z - P_{\Omega_0}) \cH_0$,
one proceeds to the net of closed, real linear subspaces
\be  \label{e.4.1}
\cL_\tau \coloneqq \{ (A - \omega_0(A)1) \Omega_0 : A = A^* \in
\fS(\cV_{o + \tau}) \}^- \subset \cK_{0} \, , \quad \tau \in V_+ \, .
\ee
Since $\Omega_0$ is cyclic and separating for $\fS(\cV_{o + \tau})$,
they are standard subspaces, 
\be \label{e.4.2}
\cL_\tau \cap \, i \, \cL_\tau  = \{ 0 \} \, , \qquad
(\cL_\tau + i \, \cL_\tau)^- = \cK_0 \, .
\ee
  Let $\Delta_0$ be the modular operator determined by
  the initial standard space $\cL_0$. It coincides with the
  restriction of the modular operator fixed by
  $(\fS_0(\cV_o), \Omega_0)$ to $\cK_0$. Equation \eqref{e.3.1}
  implies that \ 
  $\Delta_0^{-i \sigma} \cL_\tau = \cL_{e^{2 \pi \sigma} \tau}  \subset \cL_\tau$
  for $\sigma \geq 0$. Thus the inclusions $\cL_\tau \subset \cL_0$
  are half-sided modular.  
  Since $U(\sigma) \cL_\tau = \cL_{\tau + \sigma}$, $\sigma, \tau \in V_+$, 
it also follows from the spectral properties of $U$ that  $\cL_\tau$
has no non-trivial element in common with its symplectic
complement~$\cL_\tau'$, 
\be \label{e.4.3}
\cL_\tau \cap \cL_\tau' = \{0\} \, , \quad \tau \in V_+ \, .
\ee
\indent We follow now the discussion in \cite{CiLoRaRu} and define the
(real linear, unbounded) cutting projections
\be \label{e.4.4}
P_\tau : \cL_\tau + \cL_\tau' \rightarrow \cL_\tau \, , \quad \tau \in V_+ \, .
\ee
The modular operator determined by the standard subspace $\cL_\tau$
is denoted by $\Delta_\tau$. 
Given any vector state $\Phi \in (1 - Z) \cH_0$, we
proceed to
$\Phi^\perp \coloneqq (1 - P_{\Omega_0}) \, \Phi \in \cK_0$ and put 
\be \label{e.4.5}
I_\tau(\Phi) \coloneqq \text{Im} \, \langle \Phi^\perp, P_\tau \, i \,
\ln{\Delta_\tau} \,
\Phi^\perp \rangle \, , \quad \tau \in V_+ \, .
\ee
This quantity is interpreted as information which can be extracted
from $\Phi$ by measurements with observables
in $\fS(\cV_\tau)$, $\tau \in V_+$; the information contained in the
stationary state $\Omega_0$ is put equal to $0$. As has been
shown in \cite{CiLoRaRu}, this interpretation is related to
the concept of relative entropy between $\Phi$ and 
$\Omega_0$, invented by Araki \cite{Ar}. The following
result, established in \cite{CiLoRaRu}, 
describes the information on~$\Phi$ in the course
of time.
\begin{proposition}
  Let $\Phi \in (1 - P_{\Omega_0}) \cH_0$ be a vector state. \\[-9mm]
  \begin{itemize}
  \item[(a)] $I_\tau(\Phi) \in [0,\infty]$ and there is a dense
    set of vectors $\Phi$ for which this quantity is finite,
    $\tau \in V_+$. 

  \vspace*{-2mm}
  \item[(b)] $I_\tau(\Phi) \leq I_\sigma(\Phi)$ if 
    $(\tau - \sigma) \in V_+$. 

  \vspace*{-2mm}  
  \item[(c)] Let $t \mapsto \tau(t) \coloneqq (t, t \bv)$ for
    fixed $\bv$, $|\bv| < 1$. If $I_{\tau(t_0)}(\Phi)$ is finite,
    $t_0 \geq 0$, then $t \mapsto I_{\tau(t)}(\Phi)$ is
    continuous for $t \geq t_0$, decreases monotonically, and is convex.
   \end{itemize}  
\end{proposition}  

The loss of information over time, described in this
proposition, can be understood in simple terms  
in the presence of massless single particle states 
in $\cH_0$, such as the photon. There then exist corresponding 
outgoing scattering states of massless particles
in $\cH_0$ and corresponding outgoing fields. 
Denoting by~$\fA_0^\text{out}(\cV_o)$ the algebra
generated by outgoing fields that are created in 
$\cV_o$ and, similarly, by~$\fA_0^\text{out}(\cV_{o \, -})$
the algebra generated by outgoing fields that were created in the past
cone~$\cV_{o \, -}$, one has the inclusions \cite{Bu1}
\be
\fA_0^\text{out}(\cV_o) \subset \fR(\cV_o) \subset
\fA_0^\text{out}(\cV_{o \, -})' \, .
\ee
Whereas the first inclusion holds also for 
outgoing fields of massive particles in~$\cH_0$, the second
one is a consequence of locality and the fact that massless
particles propagate with the speed of light.
Thus the loss of information can be visualized 
in this case geometrically. It is a consequence of Huygens's principle
according to which the outgoing massless particles in a state,  which
were created in the past cone $\cV_{o \, -}$, will miss the future
cone $\cV_{o}$ and thus leave no observable effects there. 
It is noteworthy that by this mechanism the notorious
infrared problems in Minkowski space, caused by infinite clouds of
massless particles, disappear for observers in lightcones
\cite{BuRo}. 

\section{The arrow of time as origin of quantization}
\label{sec5}
\setcounter{equation}{0}
As we have seen,
the hypothesis of an intrinsic arrow of time is compatible with the
standard theoretical treatment of spacetime translations as a group.
However, if one is located at some spacetime point $o$,
there arise ambiguities about the action of these translations  
and the properties of the
underlying physical system in the non-accessible
past and spacelike complement $\cN_o$. They are due to 
the presence of massless
particles. In other words, one never has perfect control on
the initial data of states which would be needed for an 
exact prediction of the results of future measurements
in~$\cV_o$. The best
one can hope for are statistical statements.
They must be  based on informations, collected in 
the past and typically  formulated in classical terms,
including quantum features that have been recorded.
This leads us to the last question (IV), namely, whether the arrow
of time entails the quantum properties of operations 
that lie ahead, but are described by classical concepts. 

That this is a meaningful idea has been expounded in recent work
\cite{BuFr2,BrDuFrRe1}. It is based on a set of operations which
comply with a specific version of the
causality principle: namely, the effects of an operation in a
given space-time region become visible exclusively in its future.
As first observed by Sorkin \cite{Sorkin},
this is a strong restriction on possible operations, and its relevance
for the measurement process in relativistic quantum physics was
recently analyzed in \cite{Bostelmann, Fewster, FewsterVerch}.
We determine here, with a simple example, a consistent
choice of operations. These operations are characterized by concepts
of classical field theory. They are conceived to
describe perturbations which are caused by
adding interaction terms to a given Lagrangian. 
As we shall see, the causal constraints,
\ie the arrow of time, imply that the operations generate a non-commutative
group which, together with the dynamical constraints, leads
to the formalism of quantum field theory. 

We consider a scalar field which propagates in Minkowski space $\cM$.
Its classical configurations are real, smooth functions $x \mapsto \phi(x)$.
The dynamics of the field is described in classical terms as well
and given by relativistic Lagrangians. 
We treat here the simple case of a non-interacting field
with Lagrangian density 
\be \label{e.5.1} 
x \mapsto L(x)[\phi] = (1/2)(\partial_\mu \phi(x) \partial^\mu \phi(x)
- m^2 \phi(x)^2) \, .
\ee
Here $\partial_\mu$ is the partial derivative with regard to the
$\mu$-component of $x$ and \mbox{$m \geq 0$} is the mass of the field.
  The  variations of the corresponding  action 
  are given for real, smooth functions with compact
  support, $\phi_0 \in \cD(\RR^4)$, by
\begin{align} \label{e.5.2} 
\phi & \mapsto \delta L(\phi_0)[\phi] \coloneqq
\int \! dx \, \big( L(x)[\phi + \phi_0] - L(x)[\phi] \big) \\
& = (1/2) \int \! dx \,
\big( \partial_\mu \phi_0(x) \partial^\mu \phi_0(x) - m^2 \phi_0(x)^2 \big) 
- \int \! dx \, \big( \square \phi_0(x) + m^2 \phi_0(x) \big) \phi(x) \, ,
\nonumber 
\end{align}
where $\square$ is the d'Alembertian. It is a special functional on the fields
of the specific form
\be \label{e.5.3}
\phi \mapsto F[\phi]= c + \int \! dx \, f(x) \phi(x) \, , \quad c \in \RR \, ,
\ f \in \cD(\RR^4) \, .
\ee
These functionals are regarded as perturbations
of the dynamics. They arise by adding
to the Lagrangian~\eqref{e.5.1}
a c-number function $x \mapsto c(x)$, which
integrates to $c$, and a 
term $x \mapsto f(x) \phi(x)$ which
is linear in the field. We restrict our attention 
to pertubations of this simple form, cf.\ \cite{BuFr2} for more
general examples. The support of a functional $F$ in Minkowski
space is identified with the support of
the underlying test function $f$. The constant functional
$\phi \mapsto c[\phi] = c$ has empty support and can be 
assigned to any spacetime region. We note in conclusion that
the functions $\phi_0 \in \cD(\RR^4)$, appearing in the variations of
the action, induce shifts of the functionals.
They are denoted 
by~$\, \phi \mapsto F^{\phi_0}[\phi] \coloneqq F[\phi + \phi_0]$.

  After this outline of the
  classical input, we consider now operations which
  are labeled by the functionals $F$. 
    The symbols $S_0(F)$ denote operations, determined by $F$,
    in presence of the unperturbed dynamics. They are
    conceived to describe
    perturbations which are caused by a local change of the
    dynamics through $F$.  Similarly, the symbols $S_G(F)$ denote
    the same operations in presence of the dynamics changed by $G$. 
  According to this interpretation,
  the products of these operations (their composition)
  are assumed to satisfy for functionals
  $F, G, H$ the relation
  \be  \label{e.5.4}
  S_H(G) S_{G + H}(F) = S_H(F + G) \, , \quad S_H(0) = 1 \, .
  \ee
  It follows that the operations have an inverse,
  $S_H(G)^{-1} = S_{H+G}(-G)$. Moreover, for any choice of
  $F,G$, one has the relation $S_G(F) = S_0(G)^{-1} S_0(F +G)$,
  known as Bogoliubov-formula \cite{BoSh}.
  
  The operation $S_G(F)$ is assumed to be localized in Minkowski space
  in the support region of $F$, irrespective of the choice of $G$.
  In order to express the causality properties of these operations,
  we must compare the supports of the underlying functionals.
  We write $G \succ F$
  if $G$ is later than $F$, \ie there is some Cauchy surface
  such that $\text{supp} \, G$ lies above and $\text{supp} \, F$
  beneath it. According to the causality condition, indicated
  above, the operation $S_{G + H}(F)$ does not depend on the choice
  of the functional $H$ if it is later than $F$, \ie
  \be   \label{e.5.5}
            S_{G + H}(F) = S_G(F) \quad \text{if} \quad H \succ F \, .
  \ee

  The preceding relations suggest to consider
  the group $\cG_0$ generated by the operations $S_0(F)$ for the chosen
  Lagrangian \eqref{e.5.1}. It is characterized by
  the following three relations.
  \be    \label{e.5.6}
          S_0(F) S_0(G) = S_0(F + G) \quad \text{if} \quad F \succ G \, .
  \ee
  In this factorization condition the arrow of time enters.
  Note that the product
  of operations is not commutative, the causal order of the
  functionals matters. 
  But the relation implies that the constant
  functionals $\phi \mapsto c[\phi] = c$ determine elements
  $S(c)$ of the center of $\cG_0$. They satisfy 
  $S(c_1) S(c_2) = S(c_1 + c_2)$. Choosing some scale
  factor, one fixes these central operations and puts    
  \be  \label{e.5.7}
        S_0(c) = e^{ic} \, 1 \, , \quad c \in \RR \, . 
  \ee 
We will see later that the chosen scale factor is related to Planck's constant.

The dynamics induced by the Lagrangian $L$ imposes further
relations on the operations, put forward in \cite{BuFr2}.
They are given for  functionals $F$ and test
functions $\phi_0 \in \cD(\RR^4)$ by
\begin{equation} \label{e.5.8}
S_0(F)=S_0(F^{\phi_0}+\delta L(\phi_0)) \, .
\end{equation}
These relations correspond to exponentiated versions
of the Schwinger-Dyson equation obtained in
the algebraic approach to perturbation theory. 

By standard arguments one can proceed from the group $\cG_0$ of
operations to a C*-algebra $\fA$. It is generated by a net of
local subalgebras on Minkowski space which is
determined by operations having support in
the corresponding spacetime regions. The
algebra $\fA$ satisfies all axioms of local quantum physics, which were used in
the preceding sections, cf.\  \cite{BuFr2}.
This feature is a consequence of the arrow of time that
is incorporated into the factorization condition for operations.

In the example considered here, this assertion can be established by a
straightforward computation, based on the three 
defining equations of the group $\cG_0$ \cite{BuFr2}. One considers for
arbitrary test functions $f \in \cD(\RR^4)$ the functionals 
\be \label{e.5.9}
\phi \mapsto F_W(f)[\phi] \coloneqq
(1/2) \int \! dx dy \, f (x) \Delta _D(x - y) f (y) +
 \int \! dz \, f (z) \phi(z) \, , 
\ee
where $\Delta_D \coloneqq (1/2) (\Delta_R + \Delta_A)$
is the mean of the retarded and advanced solutions of the
Klein-Gordon equation with mass $m$. Putting
$W(f) \coloneqq S(F_W(f))$, the following relations are obtained 
for arbitrary test functions $f_1, f_2, f_3 \in \cD(\RR^4)$:
\begin{align} \label{e.5.10}
&  W(f_1) W(f_2) = e^{-(i/2) \int \! dx dy \, f_1(x) \Delta(x-y) f_2(y)}
    W(f_1 + f_2) \, , \nonumber \\
&  W((\square + m^2) f_3) = 1 \, , 
\end{align} 
where 
$\Delta \coloneqq (\Delta_R - \Delta_A)$ (Pauli-Jordan function).
Thus the operators $W(f)$ are exponentials of a
real, scalar, local quantum field of mass $m$
that satisfies the Klein-Gordon
equation and is integrated with test functions $f$ (Weyl operators).
The exponent of the phase factor in
\eqref{e.5.10} reveals that the scale chosen in
equation~\eqref{e.5.7} amounts to putting Planck's constant
equal to $1$.
Note that these equations are obtained without having
imposed any quantization rules from the outset. They emerge
from the constraints imposed by the arrow of~time.

\section{Summary}
\label{sec6}
\setcounter{equation}{0}

In the present article we have examined the
consequences of the hypothesis that the arrow of
time is a fundamental fact which enters in the evolution of
all systems.  There is no return to the past. We have clarified in
a first step how the assumption that time forms
a semigroup is related to the standard description of spacetime
transformations, forming a group.
That the latter description is consistent with the present input 
relies on the empirical fact
that experiments can be repeated, \ie one can prepare the same 
state many times. This suggests that there is some stationary
background. We have discussed
here the case of a vacuum state. As we have seen, its properties imply
that the semigroup of time translations can be unitarily implemented
in the corresponding GNS-representation. One can then proceed to
a unitary group of all spacetime translations. It allows one to move
theoretically backwards in time. Let us mention as an aside that
similar results obtain if one proceeds from a thermal
equilibrium state that satisfies the KMS~condition.

Next, we studied the question whether the extension of
the semigroup of time translations to the group of spacetime translations
is unique. It turned out that in general there arise ambiguities.
Namely, given a future lightcone, where the evolution of
operations and measurements is described by a given semigroup of
time translations, there can exist different extensions of this
semigroup which describe different dynamics in the past.
As a result, past properties of states can  not be 
reconstructed with certainty in this case. We have also
seen that if such ambiguities occur,
there exist states describing excitations of arbitrarily small mass. 

In order to clarify whether these states are responsible for
the loss of control of past properties, we made use of
a novel quantity, introduced in \cite{CiLoRaRu}. It measures the
quantum information contained in a state, relative to the vacuum.
We have shown that the information in the states
of interest here decreases
monotonically with time. Alternatively, one may
speak of an  increase of entropy. The underlying
excitations escape continuously into the non-accessible part
of Minkowski space. In case of massless particles in the
vacuum sector, this is known to be a consequence of Huygens's principle; 
but there may well exist other entities with this property.
These excitations cause dissipation
and an inevitable loss of quantum information.
What remains accessible over time
are material systems. They can carry along
information which may be expressed in classical terms
(ordinary language).
In order to exhibit their quantum features one needs to 
perform renewed operations, which will produce again
the transient excitations. 

These points were complemented in a final step by a survey
on recent results in \cite{BuFr2}, where the arrow of time 
was shown to be a source of quantization.
Given a system, its properties are described in classical terms,
which may be thought of as being based on informations
obtained in the past. One then considers
localized operations which are described by
functionals on the trajectories of the underlying
classical system. They are interpreted as  
perturbations caused by local changes of the dynamics.
There are two fundamental relations between these operations.
The first one describes the net effect of successive
operations. There the arrow of time enters.
The second relation involves the dynamics in form
of a Lagrangian. These relations determine a dynamical
group. No quantization rules were assumed from
the outset. Nevertheless these
ingredients determine concrete algebras which fit into the framework
of local quantum physics, cf.\ \cite{BuFr2, BrDuFrRe1}. 
It is worth  \mbox{mentioning} that a similar approach works also in 
case of non-relativistic quantum mechanics~\cite{BuFr3}. 

In contrast to the standard approach to
quantum physics, where the observables are in focus, the basic
ingredients are here the operations on the underlying
system. Such operations, described by unitary
operators, can be used as a substitute for observables.
As has been shown in \cite{BuSt}, they allow to 
determine with arbitrary precision
the expectation values of given basic observables
(projections) in given subspaces of states. Moreover,
after their action the states are elements of the
corresponding spectral subspaces, there is no 
collapse of wave functions. For this reason, these operations
were called primitive observables in  \cite{BuSt}. 

The present results are  surprisingly close to the view of 
Niels Bohr, who has argued that observations must be described
in ordinary language supplemented with classical physical
concepts \cite[p.\ 124]{Fa}.
What Bohr did not know at his time is the fact that their
quantum features can be traced to the arrow of time.

So, in summary, we come to the conclusion that the hypothesis
of an intrinsic arrow of time, inherent in all systems, 
is meaningful. It needs no justification by the Second Law.
As a matter of fact, it implies the increase of entropy
(loss of information) over
time, as we have seen in a simple example. The initial 
problem, that this hypothesis is in conflict with the
efficient theoretical usage of the group of spacetime translations
has a surprisingly simple solution. But this solution also reveals 
that the standard theoretical treatment is to some
extent ambiguous.

There remain, however, several questions.
Among them is the description of events that can be
regarded as secured information, cf.\
\mbox{\cite[VII,3]{Ha}} and \cite{BlFrSch}.
From our present point of view, these events provide the
basis for the formulation of future operations.
Another problem is the discussion of the arrow
of time in curved backgrounds. There the evolution of systems
can be described by a principle of local covariance \cite{BrFrVe}.
Because of the lack of stationary states it is, however, less clear how
to reconstruct from data in lightcones a consistent picture of the past.
Thus, in view of the present results, it seems \mbox{worthwhile} to take a
fresh look at the foundations of quantum physics, based on this
new~paradigm.

\noindent 
\textbf{\Large Acknowledgments}

\vspace*{1mm} \noindent 
DB is grateful to Dorothea Bahns and the Mathematics Institute
of the University of G\"ottingen for their continuing hospitality.

\end{document}